\documentclass[conference,compsoc]{IEEEtran}

\usepackage{graphicx} 
\usepackage{algorithm}
\usepackage{algorithmic}

\usepackage{amsthm}
\usepackage{amsmath}
\usepackage{amssymb}
\theoremstyle{definition}
\newtheorem{definition}{Definition}[section] 
\newtheorem{theorem}{Theorem}[section]

\DeclareMathOperator{\Hash}{Hash}
\DeclareMathOperator{\Gen}{Gen}
\DeclareMathOperator{\secret}{\textit{secret}}
\DeclareMathOperator{\AdvHash}{AdvHash}

\usepackage{hyperref}
\usepackage{xcolor}

\ifCLASSOPTIONcompsoc
 
  \usepackage[nocompress]{cite}
\else
 
  \usepackage{cite}
\fi

\ifCLASSINFOpdf
 
\else
 
\fi

\begin{document}
\title{Watermarking Generative Categorical Data}

\title{Watermarking Generative Categorical Data}

\author{
\IEEEauthorblockN{Bochao Gu\textsuperscript{*}}
\IEEEauthorblockA{Department of Mathematics\\
University of California, Los Angeles\\
Los Angeles, California 90024\\
Email: \texttt{billgubochao@g.ucla.edu}}
\and
\IEEEauthorblockN{Hengzhi He\textsuperscript{*}}
\IEEEauthorblockA{Department of Statistics\\
University of California, Los Angeles\\
Los Angeles, California 90024\\
Email: \texttt{hengzhihe@g.ucla.edu}}
\and
\IEEEauthorblockN{Guang Cheng}
\IEEEauthorblockA{Department of Statistics\\
University of California, Los Angeles\\
Los Angeles, California 90024\\
Email: \texttt{guangcheng@ucla.edu}}
}

\maketitle

\begingroup
\renewcommand\thefootnote{\textsuperscript{*}}
\footnotetext{Bochao Gu and Hengzhi He contributed equally to this work.}
\endgroup

\begin{abstract}
In this paper, we propose a novel statistical framework for watermarking generative {\em categorical} data. Our method systematically embeds pre-agreed secret signals by splitting the data distribution into two components and modifying one distribution based on a deterministic relationship with the other, ensuring the watermark is embedded at the distribution-level. To verify the watermark, we introduce an insertion inverse algorithm and detect its presence by measuring the total variation distance between the inverse-decoded data and the original distribution. Unlike previous categorical watermarking methods, which primarily focus on embedding watermarks into a given dataset, our approach operates at the distribution-level, allowing for verification from a statistical distributional perspective. This makes it particularly well-suited for the modern paradigm of synthetic data generation, where the underlying data distribution, rather than specific data points, is of primary importance. The effectiveness of our method is demonstrated through both theoretical analysis and empirical validation.
\end{abstract}

\IEEEpeerreviewmaketitle

\section{Introduction}
The rapid development of generative models has led to significant advancements in image (\cite{goodfellow2014generative,song2020score,song2023consistency}) and text generation (\cite{achiam2023gpt,dubey2024llama,li2024pre,yu2022latent}), where they have found important applications. These models are also being used to generate high-quality synthetic tabular data, opening new possibilities in this domain (\cite{xu2019modeling,kotelnikov2023tabddpm,zhangmixed,ouyang2023missdiff,kinakh2024tabular}). Models such as CTGAN (\cite{xu2019modeling}), TabDDPM (\cite{kotelnikov2023tabddpm}), and TabSyn (\cite{zhangmixed}) have demonstrated the capability to produce datasets that closely resemble real data, including finance (\cite{potluru2023synthetic}) and health care (\cite{chen2021synthetic}). However, the increasing use of AI-generated data has raised critical concerns about distinguishing it from authentic content and identifying its source. Failing to address these issues can lead to significant problems, including copyright infringement and the spread of misinformation. These concerns have driven regulatory bodies at both national and international levels to take action. Notably, recent executive directives from the White House \footnote{
\href{https://www.whitehouse.gov/briefing-room/presidential-actions/2023/10/30/executive-order-on-the-safe-secure-and-trustworthy-development-and-use-of-artificial-intelligence/}{https://www.whitehouse.gov/briefing-room/presidential-actions/2023/10/30/executive-order-on-the-safe-secure-and-trustworthy-development-and-use-of-artificial-intelligence/}} and the European Union's proposed Artificial Intelligence Act \footnote{
\href{https://artificialintelligenceact.eu/wp-content/uploads/2024/02/AIA-Trilogue-Committee.pdf}{https://artificialintelligenceact.eu/wp-content/uploads/2024/02/AIA-Trilogue-Committee.pdf}} emphasize the need for robust mechanisms to ensure that AI-generated content can be identified and traced back to its origins. 

To achieve these goals, watermarking techniques present an effective solution. Watermarks have already played a crucial role in tracing the provenance of content in image (\cite{wen2023treerings,zhao2023provable}) and text domains (\cite{kirchenbauer2023watermark,christ2024undetectable}). However, the application of watermarking techniques in generative tabular data remains limited. 

Recent studies, such as \cite{he2024watermarking,zheng2024tabularmark,tang2024ripple}, have explored watermarking techniques for generative numerical data, but there has been relatively little focus on generative categorical data. Although extensive research on watermarking techniques for categorical data exists in the database field, most studies focus on static tables or tables with fixed content. For example, the watermarking method proposed by \cite{li2004tamper} embeds the watermark by arranging rows in a specific sequence, so that the order of data rows carries the watermark. Another approach by \cite{agrawal2002watermarking} assumes each row has a unique primary key, with the watermark deterministically added based on this key. While these methods demonstrate robustness in various contexts, they may not be directly applicable in the era of generative data, where traditional watermarking techniques face new challenges. Unlike static datasets, data in the era of generative models is highly susceptible to various modifications that can disrupt embedded watermarks. These include reordering, cropping, and other similar modifications. Additionally, generative models introduce a unique challenge: they can learn the underlying distribution of the original data and use this learned structure to generate entirely new datasets that are statistically similar but fundamentally different in content, further complicating watermark preservation. 

In this scenario, ensuring that watermarks can be accurately detected—even after the data has been learned and regenerated by a generative model—presents a significant challenge. Developing watermarking techniques that can withstand these transformations and reliably verify the origin of generative data with high accuracy is therefore an open and critical problem in the field.

In this paper, we address the unique challenges of watermarking in the era of generative data by developing a distribution-level watermarking approach tailored for categorical data. Our main contributions are as follows:

\begin{itemize}
    \item \textbf{A Novel Distribution-Level Watermarking Framework}: We propose a statistical framework that systematically embeds watermark signals at the distribution-level. By strategically adjusting the original data distribution in a controlled manner, our approach enables reliable watermark embedding while maintaining the natural appearance of the synthetic data.
    
    \item \textbf{Robust Verification Through Hypothesis Testing}: To verify the presence of the watermark, we introduce an inverse decoding process and employ a statistical hypothesis testing method that measures the total variation distance between the inverse-decoded data and the original data distribution. This approach allows for robust verification of watermarks from a distributional perspective, making it resilient to transformations commonly encountered in generative models.
    
    \item \textbf{Theoretical and Empirical Validation}: We provide both theoretical analysis and empirical validation to demonstrate the effectiveness and robustness of our method. Through rigorous testing, we show that our approach reliably detects embedded watermarks even when the synthetic data undergoes regeneration, demonstrating its suitability for practical applications.
    
    \item \textbf{Post-Processing Compatibility with Various Generative Models}: Our framework is designed as a post-processing method that does not rely on any specific generative model architecture. This makes it adaptable to a wide range of generative models, allowing watermarks to be embedded and verified regardless of the underlying model used to generate the data.
\end{itemize}

The remainder of the paper is organized as follows. Section 2 introduces related work, which serves as the basis for our approach. Section 3 outlines our problem setup and watermarking scheme. Section 4 presents our empirical results, and Section 5 introduces two advanced methods that build upon the foundational techniques discussed in Section 3.

\section{Related Work}
\textbf{LLM Watermark} Recent work has explored watermarking techniques for large language models. One prominent approach is the ``green-red list" based watermark, introduced by \cite{kirchenbauer2023watermark}. This method partitions the vocabulary into ``green" and "red" lists. During each token generation step, the model is biased to sample primarily or exclusively from the green list. To detect the watermark, the proportion of green tokens in the generated text is measured—an unusually high ratio of green tokens indicates the presence of a watermark, as natural text would not typically exhibit this distribution. Subsequent works, such as \cite{zhao2023provable, li2023plmmark, fernandez2023three}, have further developed and expanded upon this approach.

Another representative approach for watermarking large language models is the Gumbel Watermark, introduced by \cite{aaronson2023watermarking}. This method leverages exponential minimal sampling technique to subtly encode watermarks within the generated text. By using a predetermined secret key to guide pseudo-random sampling, the Gumbel Watermark embeds a unique statistical signature detectable in the text output. This technique enables robust watermarking while maintaining the natural quality of the generated content. Subsequent works, such as \cite{zhao2024permute} and \cite{fu-etal-2024-gumbelsoft}, have further refined and extended this methodology.

\textbf{Image Watermark}
In recent years, watermarking generative image data has garnered significant interest. For example, \cite{wen2024tree} introduces a watermarking technique specifically for images generated by diffusion models, embedding watermarks into the initial noise vector during sampling. Similarly, \cite{fernandez2023stable} explores watermarking for images produced by diffusion models. However, unlike \cite{wen2024tree}, this approach fine-tunes the model's decoder to embed watermarks into the generated images, rather than modifying the initial noise directly.

\textbf{Tabular data watermark}
Effective watermarking techniques for generative tabular data remain a relatively underexplored area, though database researchers have developed watermarking schemes aimed at proving ownership of static tabular data, typically relying on the dataset's primary keys. The primary key, serving as a unique identifier for each sample in a dataset to ideally ensure distinctness, is commonly used in many watermarking schemes. For example, \cite{agrawal2002watermarking} developed a watermarking framework that sparsely modifies the least significant bits (LSB) of certain numerical features. The selection of which bits to modify is based on the primary key combined with a secure hash function. Such framework is robust under malicious attack such as sorting, scaling, and bit-flipping. 

Building on the watermarking framework established by \cite{agrawal2002watermarking}, \cite{li2005fingerprinting} introduced a fingerprinting approach that enables a data owner to create multiple watermarked copies of tabular datasets, each uniquely marked and distributable to different recipients. In this approach, Li defines a predefined binary sequence as a watermark, which is then used to perturb the dataset in accordance with the sequence. By predefining multiple watermark sequences, the data owner can produce distinct watermarked versions of the dataset. Building on both these frameworks, \cite{xiao2007second} proposed using both the least significant bit (LSB) and the second-LSB to sparsely perturb the dataset, allowing for more subtle modifications.
In cases where a primary key is absent, \cite{li2003constructing} developed methods to construct a virtual primary key using other numerical features of the dataset. To address potential issues of duplicate values in a virtual primary key, \cite{gort2020double} presented a framework for handling such challenges.
Additionally, watermarking schemes like \cite{li2004tamper} employ the hash value of the primary key along with a predefined binary sequence to reorder the sequence of samples, preserving the original data distribution. Such schemes are primarily aimed at localizing malicious attacks without altering the overall data distribution. However, these approaches are designed specifically for fixed tabular data and do not incorporate distribution-level considerations. Consequently, they depend heavily on primary keys, which are often impractical to preserve in data synthesizers. At the same time \cite{sebe2005noise} and \cite{ren2023robust} are watermarking schemes for continuous tabular data that do not rely on primary key, but their methods are challenging to extend to categroical data. 

Recently, \cite{he2024watermarking}, \cite{zheng2024tabularmark}, and \cite{ngo2024adaptive} proposed distribution-level watermarking schemes for continuous tabular data. While these approaches, like ours, embed watermarks at the distribution-level, they rely heavily on the continuous nature of the data and cannot be applied directly to discrete categorical data.

\section{Method}
\subsection{Problem Set Up and Notations}

Assume the data owner has collected some private categorical samples and organized them into a table, \( T_{ori} \). The empirical distribution of this table, \( T_{ori} \), is denoted by \( D_{ori} \). Meanwhile, the data owner also has access to a tabular data synthesizer, \( S \), which has been trained on \( D_{ori} \) and produces synthetic samples following a distribution denoted as \( D \). We assume that \( S \) preserves the fidelity well, so \( D \) closely resembles \( D_{ori} \).

In this paper, we assume that the data owner will only sell synthetic tables drawn from \( D \) to buyers. To prove ownership, the data owner perturbs the distribution \( D \) with a secret parameter, creating a modified distribution \( \hat{D}_{\text{secret}} \) that is computationally challenging to reverse back to \( D \) without knowledge of this secret. The data owner first samples a table \( T \) from \( D \) and then applies an insertion algorithm to transform \( T \) into a watermarked table, \( \hat{T}_{\text{secret}} \), with the underlying distribution \( \hat{D}_{\text{secret}} \). Different secrets can be assigned to different buyers.

When a suspicious table \( T' \) arises that may be an illicit copy of \( \hat{D}_{\text{secret}} \), the owner can analyze its underlying distribution \( D' \) and apply an inverse insertion algorithm using the secret. If the algorithm successfully recovers \( D \), it confirms that \( T' \) is indeed an illicit copy; otherwise, it is not. Notably, the data owner only needs \( T \) when inserting the watermark and \( T_{ori} \) during detection. In practice, the data owner only retains \( T_{ori} \) and a list of secrets assigned to various buyers.

Assume the data owner has a synthetic table represented as a matrix \( T \), with each row being an independent synthetic sample drawn from \( D \). Let \( D \) have dimension \( m + n \). We then split \( T \) column-wise into two data sets, \( T_x \) and \( T_y \). For instance, \( T_x \) includes columns 1 through \( m \), and \( T_y \) includes columns \( m + 1 \) through \( m + z \). Consequently, we can denote two distributions, \( X \) and \( Y \), from which samples in \( T_x \) and \( T_y \) are drawn. Clearly, \( X \) and \( Y \) are not independent unless there are two sets of independent features in \( D \).

A key component of our watermarking scheme is a one-way hash function, denoted by $\Hash_{\alpha}^{\secret}(\cdot)$. The domain of $\Hash_{\alpha}^{\secret}(\cdot)$ is a vector of any finite dimension, and its output is an integer between $0$ and $\alpha - 1$. The parameter $\secret$ secures this hash function, making it computationally difficult to deduce the input from the output without knowing $\secret$. In this paper, we primarily use $\Hash_{\#(Y)}^{\secret}$, where $\#(Y)$ represents the number of unique categories in the distribution $Y$. Later, in section 4.1, we will discuss how to implement such hash function in Python.

Since \( D \), \( X \), and \( Y \) are all categorical distributions with finite supports, we can always construct an injective mapping that assigns each category of a categorical distribution to a unique non-negative integer. Let \( M_D(\cdot) \), \( M_X(\cdot) \), and \( M_Y(\cdot) \) denote such mappings for distributions \( D \), \( X \), and \( Y \), respectively. Similarly, let \( M_D^{-1}(\cdot) \), \( M_X^{-1}(\cdot) \), and \( M_Y^{-1}(\cdot) \) denote the inverse mappings, which map integers in the appropriate domains back to their original categories. Notice that the range of $M_Y(\cdot)$ are precisely integers between $0$ and $\#(Y) - 1$. In our watermarking scheme, we will use a tabular table to construct these pairs of injective mappings. In practice, the data owner doesn't have specific details of all distributions, so the construction of such injective mapping will base on the tabular table. Particularly, we use $M_X(\cdot), M_X^{-1}(\cdot) =$ \texttt{Construct($T_x$)} to denote this process. However, if there exists a sample $x$ that does not appear in any row of $T_x$, we need to expand the domain to ensure that $M_X(\cdot)$ considers $x$ as a valid input. We denote this process as $(M_X(\cdot), M_X^{-1}(\cdot)).\texttt{update}(x)$.

Notice that if we sample many data points from the distribution $X$ and put each data point into our hash function, then we obtain a new distribution of the output. Specifically, we denote this new distribution as $\Hash_{\#(Y)}^{\secret}(X)$. Then, applying $M_Y^{-1}$ on $\Hash_{\#(Y)}^{\secret}(X)$, we obtained a new distribution that share the same support with $Y$. In the rest of this project, we denote
$$ \tilde{Y}_{\secret} := M_Y^{-1}(\Hash_{\#(Y)}^{\secret}(X)).$$

Let $p_w \in (0,1)$ be a parameter predefined by the data owner. This parameter controls the intensity of the watermark. A greater value of $p_w$ leads to more perturbation of $D$, while a smaller $p_w$ corresponds to less perturbation. Later, we will show that our insertion algorithm turned table $T_y$ into a new watermarked version $\hat{T}_{y, \secret}$, and the underlying distribution of $\hat{T}_{y, \secret}$ is 
\[
\begin{cases}
    \hat{Y}_{\secret} \stackrel{d}{=} \tilde{Y}^{\secret} \text{ with probability $p_w$} \\
    \hat{Y}_{\secret} \stackrel{d}{=} Y \text{ with probability $1-p_w$}.
\end{cases}
\]
If we merge the distribution $X$ with the distribution $\hat{Y}_{\secret}$ in the same way that we originally split $X$ and $Y$ from $D$, we obtain a new distribution $\hat{D}_{\secret}$.

In conclusion, we summarize all the notation in table 1. 

\begin{table}[h]
    \centering
    \begin{tabular}{{|p{3cm}|p{4cm}|}}
        \hline
        $\secret$ & a secret string only data owner has \\
        \hline
        $T_{ori}$ & The original tabular data owner has \\
        \hline
        $T$ & Table of Synthetic Data  \\
        \hline
        $\hat{T}_{\secret}$ & Watermarked table using $\secret$ \\
        \hline
        $T'$ & Suspicious table that might be an illicit copy \\
        \hline
        $D_{ori}$ & The underlying Distribution of $T_{ori}$\\
        \hline
        $D$ & Distribution of $T$\\
        \hline
        $\hat{D}_{\secret}$ & Distribution of $T_{\secret}$ \\
        \hline
        $D'$ & The underlying Distribution of $T'$ \\
        \hline
        $T_{ori,x}$, $T_{ori, y}$ & Column-wise partition of the $T_{ori}$ \\
        \hline
        $T_x, T_y$ & Column-wise partition of $T$\\
        \hline
        $T'_x, T'_y$ & Column-wise partition of $T'$ \\
        \hline
        $X$, $Y$ &  Distribution of $T_x$ and $T_y$ \\
        \hline
        $\hat{Y}_{\secret}$ & Watermarked distribution of $Y$ using $\secret$ \\
        \hline
        $\Hash_{\#(Y)}^{\secret}(x)$ &  one way hash function based on $\secret$ that output an integer \\
        \hline
        $p_w$ &  Parameter that controls how dense the watermark would be \\
        \hline
        $\hat{D}_{\secret}$ & Distribution of watermarked data \\
        \hline
        $\hat{T}_{y, \secret}$ & Watermarked version of $T_y$ \\
        \hline
        $\hat{Y}_{\secret}$ &  Distribution of $\hat{T}_{y, \secret}$\\
        \hline
        $M_D(\cdot), M_X(\cdot), M_Y(\cdot)$ & To-integer mapping of $D$, $X$, and $Y$ \\
        \hline
        $M_D^{-1}(\cdot), M_X^{-1}(\cdot), M_Y^{-1}(\cdot)$ & To-category mapping of $D$, $X$, and $Y$\\
        \hline
        \texttt{Construct}, \texttt{update} &  Function that construct $M_D(\cdot)$ and $M_D^{-1}(\cdot)$ \\
        \hline
        $\tilde{Y}^{\secret}$ & Distribution used to generate $\hat{Y}_{\secret}$\\
        \hline
        
    \end{tabular}
    \caption{Notations}
    \label{tab:sample_table}
\end{table}

\subsection{Watermark Insertion}

\begin{algorithm}[H]
    \caption{Inserter$(T; T_{ori}, p_w, \Hash_{\#(Y)}^{\secret})$}
    \begin{algorithmic}[1]
    
    \STATE Construct a Bernoulli variable $B$ independent of the dataset, with $\Pr(B = 0) = p_w$.
    \STATE Split $T$ column-wise into $T_x$ and $T_y$
    \STATE Split $T_{ori}$ column-wise in the same way to form $T_{ori, x}$ and $T_{ori, y}$
    \STATE Let $M_Y(\cdot), M_Y^{-1}(\cdot) = $\texttt{Construct($T_{ori,y}$)}
    \FOR{Each sample $(x_i,y_i)$ in the $T$}
        \STATE Draw a sample $b_i$ from $B$
        \IF{$b_i = 0$}
            \STATE change $y_i$ to $M_Y^{-1}(\Hash_{\#(Y)}^{\secret}(x_i))$
        \ENDIF
    \ENDFOR
    \STATE Denote the preturbed $T_y$ as $\hat{T}_{y,\secret}$. Merge $T_x$ and $\hat{T}_{y,\secret}$ to obtain $\hat{T}_{\secret}$ 
    \STATE Return $\hat{T}_{\secret}$
    \end{algorithmic}
\end{algorithm}

If we model each row of $T_y$ being an independent sample drawing from $Y$, then each row of $\hat{T}_{\secret}$ follows distribution $\hat{Y}_{\secret}$ we have defined in the previous section. 

In practice, if the data owner has a list of secrets $\secret_1, \secret_2, \dots, \secret_n$ and wants to sell a list of synthetic tables $T_1, T_2, \dots, T_n$ to $n$ different buyers, they only need to construct $M_Y(\cdot)$ and $M_Y^{-1}(\cdot)$ once. The watermark can then be inserted separately for each $T_i$.

Note that we use \( T_{ori,y} \) to construct \( M_Y(\cdot) \) and \( M_Y^{-1}(\cdot) \), but we call \( M_Y^{-1}(\cdot) \) to replace certain samples in \( T_y \). Therefore, we require \( Y \) and the underlying distribution \( Y_{ori} \) to have identical support. In other words, the data synthesizers \( S \) must neither create any new rows in \( T_y \) that do not appear in \( T_{ori,y} \) nor delete any rows from \( T_{ori,y} \), which would lead to a missing category in \( T_y \).

\subsection{Watermark Detection}

The watermark detection algorithm is more complicated than the insertion algorithm. When analyzing a table $T'$ suspected to be an illicit copy of $\hat{T}_{\secret}$, we need the following three steps. The first step is the \textbf{construction of probability vectors.} We  construct probability vectors for both $D'$ and $D$. This involves creating the injective mappings $M_D^{\secret}(\cdot)$ and $M_D^{\secret, -1}(\cdot)$. The second step is using \textbf{insertion inverse algorithm}. As the name suggests, the insertion inverse can use the output distribution from insertion algorithm and recover the input distribution $D_{inv}$. The third step is \textbf{hypothesis testing.} We compare $D$ and $D_{inv}$. If these distributions are highly similar, we conclude that $T'$ is an illicit copy of $\hat{T}_{\secret}$.

\textbf{Construction of Probability Vectors}

In the original distributions $X$ and $Y$, some x in the support of $X$ and y in the support of $Y$ may have $\Pr(Y = y|X = x) = 0$. However, it is possible that $\Pr(\tilde{Y}_{\secret} = y | X = x) > 0$. The insertion algorithm may create new pairs $(x,y)$ that never appeared in $T$. Theoretically, the inserter cannot eliminate any category from $D$ since the probability of each category changing is exactly $p_w$. However, in practice, with finite samples in $T$, it is possible that a unique row is selected for modification every time it appears, causing it to be absent from $\hat{T}_{\secret}$. This scenario becomes more likely as $p_w$ increases. This is the reason that we don't construct $M_D(\cdot)$ and $M_D^{-1}(\cdot)$ at insertion time. Then, we first construct $M_D^{\secret}(\cdot), M_D^{\secret}(\cdot)$ which could accommodate all categories of both $D$ and potential $\hat{D}_{\secret}$. This operation necessitates assigning zero probability to certain categories in both $D$ and $\hat{D}_{\secret}$.

\begin{algorithm}[H]
    \caption{\texttt{ConstructD}($ T_{ori}, T', M_Y(\cdot), \Hash_{\#(Y)}^{\secret}(\cdot)$)}
    \begin{algorithmic}[1]
    
    \STATE $M_D(\cdot),  M_D^{-1}(\cdot) =$ \texttt{Construct($T_{ori}$)}
    \FOR{Each $d$ in the domain of $M_D(\cdot)$}
        \STATE Split $d$ into $x$ and $y$ the same way we split $T$ into $T_x$ and $T_y$
        \STATE Let $y' = \Hash_{\#(Y)}^{\secret}(x)$
        \IF{$d' = (x,y')$ is not in the support of $M_D^{\secret}(\cdot)$} 
            \STATE $(M_D^{\secret}(\cdot), M_D^{\secret,-1}(\cdot))$.\texttt{update($d'$)}
        \ENDIF
    \ENDFOR
    \STATE Return $M_D^{\secret}(\cdot), M_D^{\secret,-1}(\cdot)$
    \end{algorithmic}
\end{algorithm}

With $M_D^{\secret}(\cdot), M_D^{\secret,-1}(\cdot)$, we can construct probability vector $\text{Vec}_{\secret}(D) := [d_0, d_1, \cdots, d_{\#(D)-1}]$  $D$ such that
$$\Pr(M_D^{\secret}(D) = i) = d_i.$$
Here, $\#(D)$ denotes the size of the domain of $M_D^{\secret}(\cdot)$, which may exceed the support size of $D$. 

\textbf{Insertion Inverse Algorithm}

For all $x$ in the support of $X$ and for all $y$ in the support of $Y$, we have the follolwing prperties.

If $M_Y^{-1}(\Hash^{\secret}_{\#(Y)}(x)) \neq y$, then 
\[\Pr(\hat{Y}_{\secret} = y| X= x) = (1- p_w) \Pr(Y= y | X= x).\]

This implies 
\begin{align*}
\Pr(Y = i |X= x) = \frac{\Pr(\hat{Y}_{secret} = y | X= x)}{1-p_w}
\end{align*}

If $M_Y(\Hash^{\secret}_{\#(Y)}(x)) = y$, then 

\begin{align*}
\Pr(\hat{Y}_{\secret} = y | X= x) = \Pr(Y= y |X= x)  \\
  \quad \quad  +p_w \Pr(Y \neq y| X=x). 
\end{align*}

This implies
\begin{align*}
\Pr(Y = i |X = x) = \frac{\Pr(\hat{Y}_{\secret} = y | X= x) -p_w}{1 - p_w}
\end{align*}

Notice that $\Pr(X = x, Y = i) = \Pr(Y = i | X = x) \cdot \Pr(X = x)$. As a result, as long as the data owner has the probability vector of $D'$ along with $M_D^{\secret}$, $M_Y^{\secret}$, and $M_X^{\secret}$, they can obtain the probabilities $\Pr(X' = x, Y' = y)$, $\Pr(X' = x)$, and $\Pr(Y' = y)$ for all $x$ in the support of $X$ and $y$ in the support of $Y$. Therefore, it is straightforward to calculate all $\Pr(Y' = y | X' = x)$.

Then, we summarize the insertion inverse algorithm as follows.

\begin{algorithm}[H]
    \caption{$Inserter^{-1}(\text{Vec}_{\secret}(D'); p_w, M_D^{\secret}(\cdot)$ $M_X^{\secret}(\cdot), M_Y^{\secret}(\cdot), \Hash_{\#(Y)}^{\secret }(\cdot))$}
    \begin{algorithmic}[1]
        \FOR{Each $x$ in support of $X$ and $y$ in the support of $Y$}
            \STATE Calculate $\Pr(Y' = i | X'= x)$ 
            \IF{$M_Y(\Hash^{\secret}_{\#(Y)}(x)) = y$}
                \STATE$\Pr(Y_{inv} = y |X_{inv} = x) = \frac{\Pr(Y' = y | X'= x) -p_w}{1 - p_w}$
            \ELSE
                \STATE $\Pr(Y_{inv} = y |X_{inv}= x) = \frac{\Pr(Y' = y | X'= x)}{1-p_w}$
            \ENDIF
            \STATE $\Pr(Y_{inv} = y, X_{inv}= x) = \Pr(Y_{inv} = y |X_{inv}= x) \cdot \Pr(X= x)$
        \ENDFOR
        \STATE Return the distribution $D_{inv} = (X_{inv}, Y_{inv})$
    \end{algorithmic}
\end{algorithm}

\textbf{Hypothesis Testing}

In our detection mechanism, we use the total variation distance to measure the similarity between two distributions that share the same support. 

\begin{definition}
Let $P$ and $Q$ be two categorical distributions on the same support $\mathbb{X}$. The total variation distance between $P$ and $Q$ denoted as $d_{TV}(P, Q)$ is half of the $L_1$ norm of their probability mass function
$$d_{TV}(P, Q) = \frac{1}{2} \sum_{x \in \mathbb{X}} |P(x) - Q(x)|.$$

\end{definition}

The data owner must predefine a prior distribution $Prior$, from which probability vectors of size $\#(D)$ can be sampled. This prior distribution is used in a hypothesis test to determine whether a table $T'$ is an illicit copy. 

$H_0$: The underlying distribution of $T'$ is a sample from $Prior$. 

$H_a$: The underlying distribution of $T'$ is $\hat{D}_{\secret}$.

Finally, we summarize the complete watermark detection algorithm that incorporate construction of probability vectors and insertion inverse algorithm below.

\begin{algorithm}[H]
    \caption{$Detector(T', T_{ori};pw, \Hash_{\#(Y)}^{\secret}(\cdot)), Prior)$}
    \begin{algorithmic}[1]
    
    \STATE Split $T'$ into $T'_x$ and $T'_y$ the same way we split $T_{ori}$ into $T_{ori,x}$ and $T_{ori,y}$
    \STATE $(M_Y(\cdot), M_Y^{-1}(\cdot))= $ \texttt{Construct($T_{ori,y}$)}
    \STATE $(M_D^{\secret}(\cdot),M_D^{\secret, -1}(\cdot))$ = \\ \texttt{ConstructD}$(T_{ori}, T', M_Y(\cdot), \Hash_{\#(Y)}^{\secret}(\cdot))$
    \STATE Use $T'$, $T$, and $M_D^{\secret}(\cdot)$ to build $\text{Vec}_{\secret}(D')$ and $\text{Vec}_{\secret}(D)$
    \STATE  $D_{inv} = Inserter^{-1}(\text{Vec}_{\secret}(D'); p_w, M_D^{\secret}(\cdot), \Hash_{\#(Y)}^{\secret }(\cdot))$
    \STATE Use $\text{Vec}_{\secret}(D')$ and $\text{Vec}_{\secret}(D)$ to calculate $d = d_{TV}(D, D')$
    \STATE Sample many probability vectors from $Prior$
    \STATE
    \FOR{Each probability vectors $D_{sam}$ we have sampled}
        \STATE $D_{sam, inv} = Inserter^{-1}(\text{Vec}_{\secret}(D_{sam});$ \\$ p_w, M_D^{\secret}, \Hash_{\#(Y)}^{\secret }(\cdot))$
        \STATE $d_{sam} = d_{TV}(D_{sam}, D')$
    \ENDFOR
    \STATE
    \STATE Let p-value be the proportion of $d_{sam}$ such that $d_{sam} \le d$
    \IF{p-value is smaller than a predefined siginifcance level}
        \STATE Reject $H_0$ and conclude $T'$ is an illicit copy 
    \ELSE
        \STATE Fail to reject $H_0$ and conclude is not an illicit copy
    \ENDIF
    \end{algorithmic}
\end{algorithm}

\subsection{Analysis of Watermarking Scheme}
One natural question to ask is whether the watermarking scheme preserves the utility of $T$. While demonstrating the utility of a table depends on the specific downstream task, it is challenging to design a comprehensive evaluation of the utility of a watermarking scheme. However, since the insertion of the watermark is sparse and controlled by the parameter $p_w$, $\hat{D}_{\secret}$ closely resembles $D$. Specifically, we have developed the following theorem to bound the distributional shift introduced by our watermarking scheme.

\begin{theorem}
    Let $T$ be the unwatermarked table and $\hat{T}_{\secret}$ be the watermarked table using our watermark insertion algorithm. Let $D$ and $\hat{D}_{\secret}$ be the underlying distributions of $T$ and $\hat{T}_{\secret}$. Then we have 

    $$d_{TV}(D, \hat{D}_{\secret}) \le p_w$$
\end{theorem}

\begin{proof}
    Let $k := \#(D) - 1$ and let $[d_0, d_1, \cdots, d_k]$ be the probability vector of $D$ such that 
    $$d_i = \Pr(M_D^{\secret}(D) = i), \ 0 \leq i \leq k.$$ 
    Similarly, we can define $[\hat{d}_0, \hat{d}_1, \cdots, \hat{d}_k]$ as the probability vector of $\hat{D}_{\secret}$

    Then, define  
    \[
    \begin{cases}
        p_i  = \Pr(M_D^{\secret}(\hat{D}_{\secret}) = i | M_D^{\secret}(D) \neq i) \\
        q_i = \Pr(M_D^{\secret}(\hat{D}_{\secret}) \neq i | M_D^{\secret}(D) = i) .
    \end{cases}
    \]

    It is easy to show that 
    $$ \sum_{i = 0}^{k} p_i = \sum_{i = 0}^{k} q_i \leq p_w$$

    Then, we have 

    \[
    \hat{d}_i = d_i +p_i - q_i. 
    \]

    Then, 
    \begin{align*}
        d_{TV}(Y, \hat{Y}) &= \frac{1}{2} \sum_{i = 0}^{k} \left| d_i - \hat{d}_i\right| \\
        &= \frac{1}{2}  \sum_{i = 0}^{k} \left| p_i - q_i\right|\\
        &\leq \frac{1}{2} \sum_{i = 0}^{k} |p_i| + |-q_i| \\
        &= \frac{1}{2} \left(\sum_{i = 0}^{k}|p_i| + \sum_{i = 0}^{k}|q_i| \right)\\
        &\leq p_w
    \end{align*}
    This completes the proof
    
\end{proof}

Another important question is whether a white-box adversarial attacker could replicate the detection algorithm, potentially enabling them to remove the watermark. Fortunately, due to the discontinuity property of one-way hash functions, it is computationally challenging to predict the output of $\Hash_{\#(Y)}^{\secret}(x)$ for any valid $x$ without knowledge of $\secret$. We formalize this intuition in the following observation.

Let $SECRET$ be a large set of strings, each of which can serve as a possible $\secret$. If we treat $SECRET$ as a uniform distribution, where each $\secret \in SECRET$ has an equal probability of being selected, then for any fixed $x$, we can view $\Hash_{\#(Y)}^{\secret}(x)$ as a random variable over the support of $Y$. As a result, $\Hash_{\#(Y)}^{\secret}(x)$ follows a uniform distribution.

While a rigorous proof would require detailed analysis of the properties of one-way hash functions and $SECRET$, we have empirically demonstrated that this assumption is valid.

We construct a big set $SECRET$ that contains many different strings. In the experiment, we first set $\#(Y)$ to be 3, 5, and 10 and we sample $12000$ strings without replacement from $SECRET$ and draw the distribution in following figures. Notice that the distribution are close to uniform distribution, which validates our assumption.
\begin{figure}[H]
    \centering
    \includegraphics[width=0.35\textwidth]{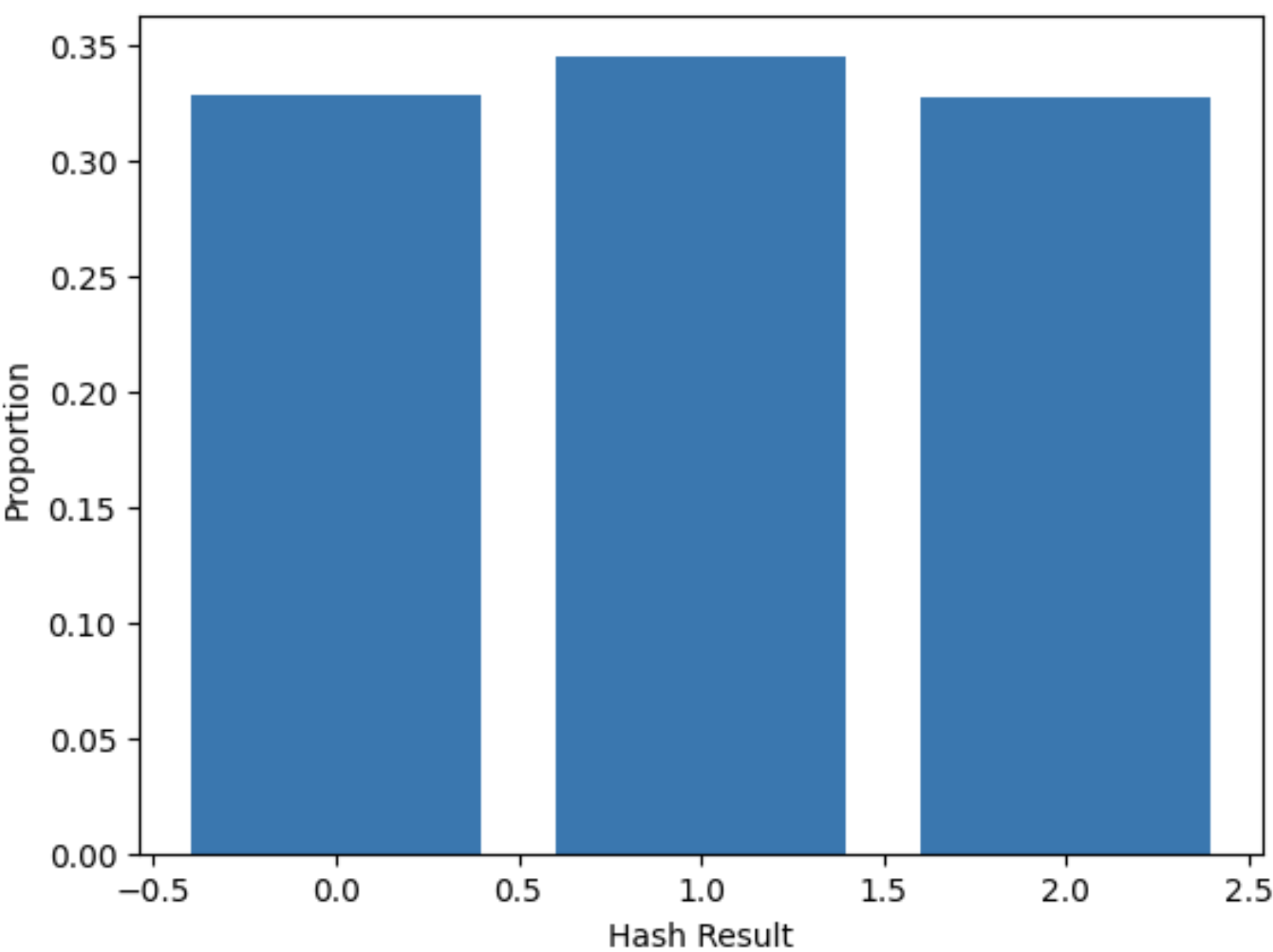}
    \caption{$\#(Y) = 3,$ samples 12000 $\secret$, $x =[2,2,2,23]$}
    \label{fig:beautiful}
\end{figure}

\begin{figure}[H]
    \centering
    \includegraphics[width=0.35\textwidth]{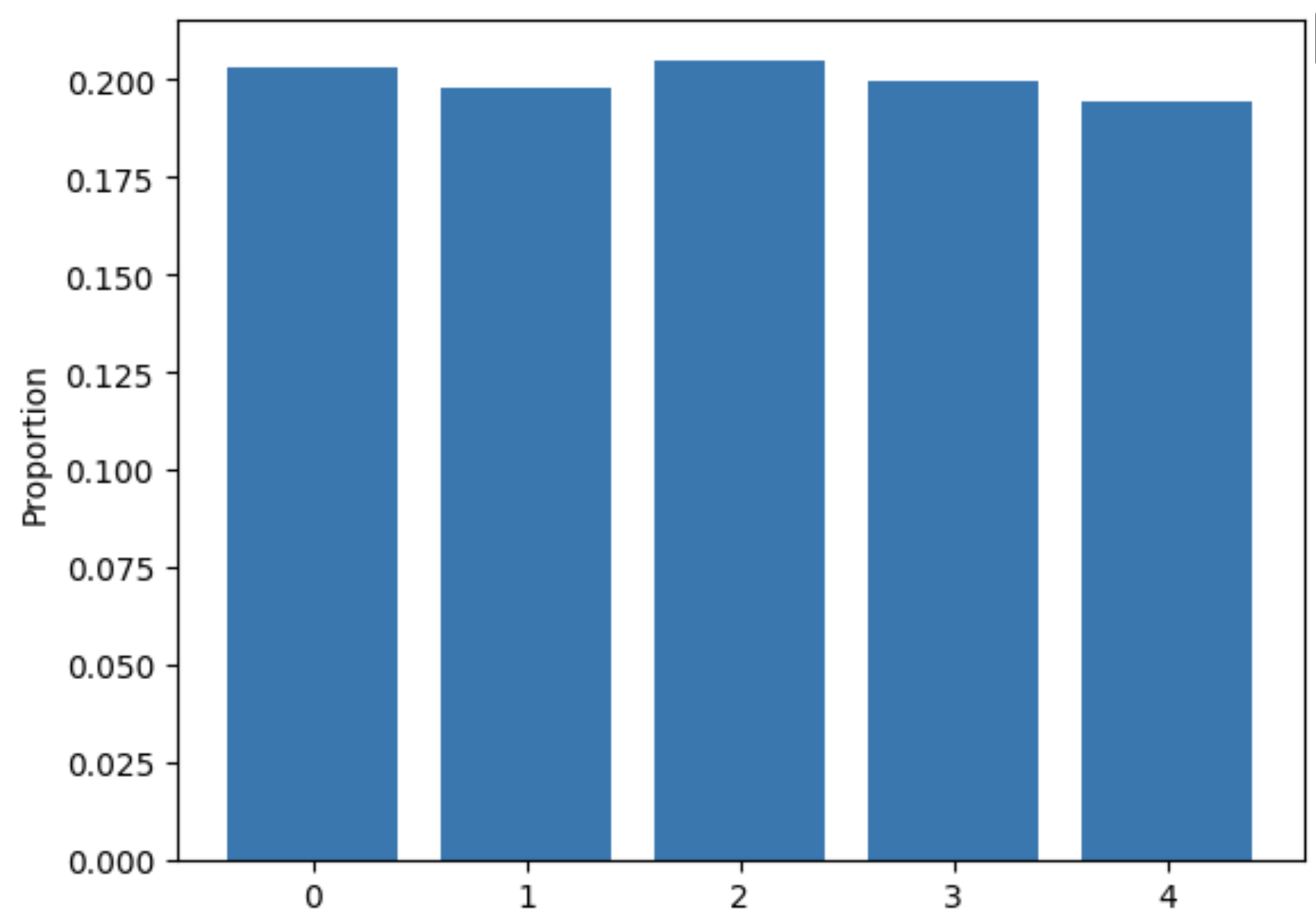}
    \caption{$\#(Y) = 5,$ samples 12000 $\secret$, $x= [2,24,2,23]$}
    \label{fig:beautiful}
\end{figure}

\begin{figure}[H]
    \centering
    \includegraphics[width=0.35\textwidth]{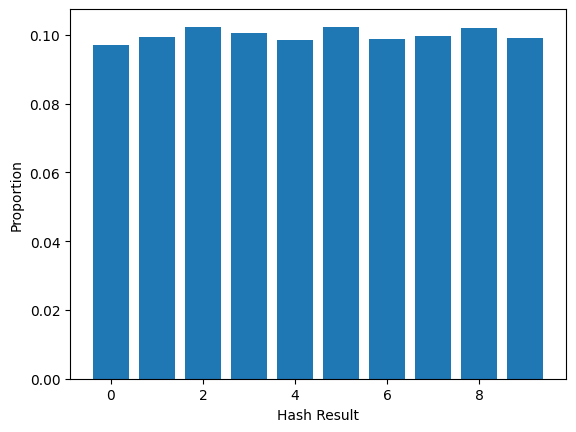}
    \caption{$\#(Y) = 10,$ samples 12000 $\secret$, $x= [233,2,2,23]$}
    \label{fig:beautiful}
\end{figure}

Then, to replicate this process, the adversary needs to correctly guess the output of $\Hash_{\#(Y)}^{\secret}(x)$ of all $x$ in the domain. With the observation above, the probability of adversary to correctly guess all output is $\frac{1}{\#(Y)^{\#(X)}}$, which is negligible when the data set has relatively high entropy. Meanwhile, since we use a bernoulli variable $B$ that is independent with our dataset to select places where watermark is inserted, it's also impossible for adversary to correctly guess the place where we change $T_y$.

In the traditional setup of watermarking problems, the owner has only one table $T_{ori}$ and a set of binary sequences $WM_1, WM_2, \cdots, WM_n$ that serve as watermarks to be inserted into $T_{ori}$. Then, the data owner distributes $n$ watermarked tables, all generated from this single $T_{ori}$. When two or more buyers compare their tables, they can identify locations where differences occur. These differences arise from the different watermark bits inserted at those positions. Consequently, the data owner faces the risk that these buyers can collectively identify all locations where watermarks have been inserted. In our case, each buyer receives different watermarked tables that are generated from different base tables. Additionally, the locations where we perturb the dataset are determined by an independent Bernoulli variable, and the watermark is implemented at the distribution-level. This makes it highly challenging for buyers to identify perturbation locations in our watermarking scheme, even if they collaborate.

Another advantage of our distribution-wise watermarking scheme is that the watermark of $\hat{T}_{\secret}$ is preserved even when buyers use it to train other synthesizers and generate new tables. We will demonstrate this property in Section 4. In contrast, traditional watermarking schemes that rely on primary keys to determine watermark bit placement face significant challenges in preserving watermarks in newly synthesized tables. This limitation arises primarily because popular tabular data synthesizers, such as TabDDPM in \cite{kotelnikov2023tabddpm} and TABSYN in \cite{zhangmixed}, require one-hot encoding for categorical features. Since primary keys contain unique values for each sample, the number of categories becomes extremely large, resulting in impractically long one-hot encoding vectors. This makes it nearly impossible for tabular synthesizers to preserve the primary key structure. Even if we use intense computation power to synthesize the primary key, the synthesizers will just draw samples from a distribution of the primary key. Since the primary key is usually something like an ID number for many data sets, it is unreasonable to assume any strong correlation or dependence between the primary key and any other features. Therefore, it is impossible for the sythesized primary key to preserve any structure that can be used in watermark detection. Meanwhile, for watermarking scheme that depends on using virtual primary key such as \cite{li2003constructing}, the construction of primary key of each sample depends on the precise values of some feautures of the table. However, the synthesizers will resample new values from the distribution of the columns and the precise values of each original sample may not be preserved anymore. 

Meanwhile, since our watermarking scheme is on the distribution-level, it is immune to any adversary attack that does not perturb the distribution such as deletion and shuffling.

\section{Experiments}
\subsection{Implementation Details}

\textbf{Hash function:} In this project, we use \texttt{hashlib.sha1()} in Python as the primary hash function, on which we build $\Hash_{\#(Y)}^{\secret}$. The output of the hash function is a binary sequence, which we can interpret as an integer. The input of $\texttt{hashlib.sha1()}$ can be any byte-like structure. In this project, all experiments are conducted using Python 3.8.8, with data stored in NumPy arrays of dtype \texttt{int64}. Using the \texttt{tobytes()} attribute of NumPy, we can efficiently convert data into byte-like structures. Similarly, Python’s \texttt{encode()} method allows us to convert strings into byte-like structures. This enables us to securely concatenate a secret string with  input data, creating a hash function that can only be accessed by the inserter and detector.  We then compute the modulus of the \texttt{hashlib.sha1()} output with $\#(Y)$ to obtain an integer value between $0$ and $\#(Y) - 1$.     

\textbf{Injective Mapping:}  The simplest approach to construct a pair of injective mappings, such as \( M_X(\cdot) \) and \( M_X^{-1}(\cdot) \), is to create a pair of look-up tables. In Python, we can achieve this by using two dictionaries. Meanwhile, it is easy for users to add new keys into dictionaries, so the \texttt{update()} function is easy to implement. It is essential that this construction of dictionary pairs is deterministic, as the data owner must construct \( M_X(\cdot) \), \( M_Y^{-1}(\cdot) \), \( M_D(\cdot) \), and \( M_D^{-1}(\cdot) \) at insertion time and be able to replicate the same mappings at detection time. In particular, the construction process must not depend on any specific ordering of samples. This means that if an attacker shuffles the order of entries and creates a new table, the data owner should still be able to construct the same mappings correctly. The easiest way to ensure this is by using the \texttt{np.unique()} function, which outputs a matrix containing all unique rows of a table in an deterministic way. 

\textbf{Prior Distribution:} In the experiment, we mainly use the Dirichlet distribution as our prior. Specifically, for any probability vector of dimension $k$, we have 
$$f(p_0, \cdots, p_{k-1}; \alpha_0, \cdots, \alpha_{k-1}) = \frac{1}{\text{B}(\alpha)} \prod_{i=0}^{K-1} p_i^{\alpha_i - 1},
$$
where $\text{B}(\alpha)$ is a multivariate beta function. Also, if $\alpha_{sum} = \sum_{i =0}^{k-1} \alpha_i$, then

$$\text{Var}[p_i] = \frac{\alpha_i/ \alpha_{sum}(1- \alpha_i/ \alpha_{sum})}{\alpha_{sum} + 1}$$

Notice that if we choose smaller values for the parameters $\alpha_i$, the variance of  $p_i$ increases. Consequently, the probability vectors sampled from the Dirichlet distribution will be more widely spread across the probability simplex.

The Dirichlet distribution is an effective prior for sampling probability vectors because it generates values that sum to one, allowing flexible control over the concentration and spread of probabilities across categories. Additionally, it is conjugate to the multinomial distribution, making it convenient for Bayesian updating in categorical data models.

\subsection{True positive Rate}
To assess the accuracy of our watermarking scheme, we simulated the distribution $D$ onto which we will embed the watermark. In the first simulation, we set $X$ and $Y$ to be independent of each other. Specifically, $X = [X_1, X_2]$, where $X_1$ and $X_2$, and $X_3$ are identical and independently distributions. Each component of $X$ has support ${0,1,2,3,4,5}$ and follows a probability vector of $[0.526, 0.263, 0.053, 0.053, 0.053, 0.053]$. Similarly, $Y = [Y_1, Y_2]$, with both $Y_1$ and $Y_2$ being identical and independent distributions, following a probability vector of $[0.077, 0.077, 0.077, 0.769]$. We draw 10,000 samples from $D$ to construct our table $T$. In the detector, we configure the Dirichlet distribution with a parameter vector where each element is set to $0.1$ and has a length equal to $M_D^{\secret}(\cdot)$. We sample $500$ probability vectors from the Dirichlet distribution. Then, we set $p_w$ to be $0.05, 0.1, 0.15$ and $0.3$ respectively.

\begin{table}[h]
    \centering
    \begin{tabular}{{|p{1.5cm}|p{1.5cm}|p{1.5cm}|p{1.7cm}|}}
    \hline
    $p_w$ & $\#(D)$ & p value & $d_{TV}(D_{inv}, D)$\\
    \hline
    $0.05$ & $559$ & $0$ &  $0.0047$\\
    \hline
    $0.10$&  $559$ & $0$ & $0.0067$ \\
    \hline
    $0.15$&  $559$ & $0$ & $0.0073$ \\
    \hline
    $0.3$&  $559$ & $0$ & $0.014$ \\
    \hline
    \end{tabular}
    \caption{Simulation 1}
    \label{tab:sample_table}
\end{table}

The second simulation is quite similar to the first one except that $X = [X_1, X_2, X_3]$ and $Y = [Y_1, Y_2]$, where each column of $X$ and each column of $Y$ all follows the same marginal distribution with that of simulation 1. Keeping everything else the same, we have the following result. 
\begin{table}[h]
    \centering
    \begin{tabular}{{|p{1.5cm}|p{1.5cm}|p{1.5cm}|p{1.7cm}|}}
    \hline
    $p_w$ & $\#(D)$ & p value & $d_{TV}(D_{inv}, D)$\\
    \hline
    $0.05$ & $2322$ & $0$ &  $0.0099$\\
    \hline
    $0.10$&  $2336$ & $0$ & $0.014$ \\
    \hline
    $0.15$&  $2336$ & $0$ & $0.017$ \\
    \hline
    $0.3$&  $2343$ & $0$ & $0.027$ \\
    \hline
    \end{tabular}
    \caption{Simulation 2}
    \label{tab:sample_table}
\end{table}

In the third simulation, we generate a distribution where columns are not independent. Specifically, let $X_1$ follow a uniform distribution on the support $\{0,1,2,3\}$ and let $\epsilon$ be an independent distribution on the support $\{0,1\}$ with probability vector $[0.8, 0.2]$. Then, $X_2 \stackrel{d}{=} X_1$ with probability $0.9$ and $X_2 \stackrel{d}{=} \epsilon$ with probability $0.1$. Additionally, $X_3 = \min(3, X_1 + \epsilon)$. $X_4$ is independent of $X_1$, $X_2$, and $X_3$ and follows a uniform distribution on the support $\{0,1,2\}$. We define $X = [X_1, X_2, X_3, X_4]$. Furthermore, $Y_1 \stackrel{d}{=} X_4 + \epsilon$ with probability $0.9$ and $Y_1 \stackrel{d}{=} \epsilon$ with probability $0.1$. Finally, $Y_2 \stackrel{d}{=} Y_1$ with probability $0.9$ and $Y_2 \stackrel{d}{=} \epsilon$ with probability $0.1$, and we define $Y = [Y_1, Y_2]$. Keeping all other parameters identical to simulations 1 and 2, we obtain the following results.

\begin{table}[h]
    \centering
    \begin{tabular}{{|p{1.5cm}|p{1.5cm}|p{1.5cm}|p{1.7cm}|}}
    \hline
    $p_w$ & $\#(D)$ & p value & $d_{TV}(D_{inv}, D)$\\
    \hline
    $0.05$ & $122$ & $0$ &  $0.0081$\\
    \hline
    $0.10$&  $122$ & $0$ & $0.014$ \\
    \hline
    $0.15$&  $125$ & $0$ & $0.014$ \\
    \hline
    $0.3$&  $125$ & $0$ & $0.022$ \\
    \hline
    \end{tabular}
    \caption{Simulation 3}
    \label{tab:sample_table}
\end{table}

As a result, all three simulations demonstrate that our watermarking scheme has a highly reliable true positive rate. This is because the insertion inverse algorithm is quite robust. Even when the watermark probability $p_w$ is relatively high, the insertion inverse is still able to transform $\hat{D}_{\secret}$ back to $D_{inv}$, which is highly similar to the original $D$. Particularly, the insertion inverse algorithm produces more accurate results when the number of categories in $D$ is not excessively large. This is to be expected, as we rely on the empirical distribution of $\hat{D}_{\secret}$, and it becomes more challenging for the empirical distribution to faithfully capture the true distribution as the number of categories increases.

Simulation 4 tests whether we can still detect the watermark if buyers use tabular synthesizers to generate a new table. Let $X_1$ and $X_2$ be identical and independent uniform distributions on the support $\{0,1,2\}$. Additionally, let $\epsilon_1$ and $\epsilon_2$ be two identical and independent distributions on the support $\{0,1\}$ with probability vector $[0.95,0.05]$. Define $X_3 = X_1 - X_2 + \epsilon_1$. Further, let $Y_1$ be an independent Bernoulli distribution with probability vector $[0.5, 0.5]$, and $Y_2 = \max(1, X_1 - \epsilon_2)$.
We sample 5000 instances from this joint distribution to form the original table $T$, and set the watermark probability $p_w = 0.05$ to obtain the watermarked table $\hat{T}{\secret}$. Then, we use TabSyn from \cite{zhangmixed} to generate the attacked table $T_{attack}$ with 4500 samples using $\hat{T}{\secret}$. Finally, we apply the detector to $T_{attack}$.
TabSyn preserves the fidelity relatively well, as evidenced by the total variation distance between $T_{attack}$ and $\hat{T}{\secret}$ being only $0.069$. Additionally, the total variation distance between $D_{inv}$ and $D$ is 0.056, which yields a p-value of $0$.

Another aspect of testing the true positive rate is the robustness of our watermarking scheme. In the previous section, we have argued that our watermarking scheme is immune to any attack that doesn't change the distribution. However, if the attacker were to randomly replace some samples of the table with samples from a uniform distribution sharing the same support, then the attacker would have perturbed the distribution, which might affect the effectiveness of our watermarking scheme. We call this the replacement attack.
Let's denote the probability of the attacker replacing an original sample with a new sample as $\beta$. Let $d^0 = [d_0, d_1, \cdots, d_n]$ denote the probability vector of $\hat{D}_{\secret}$. Additionally, let $\epsilon = [\frac{1}{n}, \frac{1}{n}, \cdots, \frac{1}{n}]$ denote the probability vector of an independent uniform distribution on the same support.
Then, after one round of attack, the probability vector of the new attacked distribution $D_{attacked, 1}$ is:
$$ d^1 = (1- \beta)d^0 + \beta \epsilon.$$
Notice that the attacker may perform multiple attacks, so we obtain the probability vector of $D_{attacked, m}$ as:
$$d^m = (1- \beta)d^{m-1} + \beta \epsilon.$$

In simulation 5, we simulate this replacement attack for $800$ rounds of attack. We set $\beta = 0.98$, and Dirichlet distribution with parameter vector where each element is set to $0.03$ and has a length equal to $M_D^{\secret}$. We draw $500$ samples from the Dirichlet distribution for hypothesis testing. Specifically, we have $X_1$ and $X_2$ being identical and independent distributions on the support $\{0,1,2,3,4\}$ with probability vector $[0.375, 0.25,  0.125, 0.125, 0.125]$. $X_3 \stackrel{d}{=} \min(6, X_1 + X_2)$ with probability $0.9$, and $X_3 = 0$ with probability $0.1$. $X_4\stackrel{d}{=} X_1$ with probability $0.8$, $X_4  \stackrel{d}{=} X_2$ with probability $0.1$, and $X_4  \stackrel{d}{=} X_3$ with probability $0.1$. $X_5\stackrel{d}{=} \max(0,X_4 - X_1)$ with probability $0.9$ and $X_5 = 5$ with probability $0.1$. Define $noise_1$ to be independent distribution on support $\{0,1,2\}$ with probability vector $[0.92, 0.04, 0.04]$. Similarly define $noise_2$ to be independent distribution on support  $\{0,1\}$ with probability vector $[0.8,0.2]$. Then, define
\[
\begin{cases}
    Y_1 = \min(7, X_1+X_2 + noise_1) \\
    Y_2 = \min(7, X_3 + X_4 + noise_1) \\
    Y_3 = \max(0,Y_2 - X_5 - noise_1) \\
    Y_4 = \min(0,Y_2- X_4+ noise_2) \\
\end{cases}
\]
We draw $50,000$ samples from this distribution and set $p_w$ to be $0.05$ and $0.1$ respectively. Let $D_{attack,m}^{-1}$ denotes the distribution obtained after puting $D_{attack,m}$ into the insertion inverse algorithm. We obtain the following result.

\begin{figure}[H]
    \centering
    \includegraphics[width=0.35\textwidth]{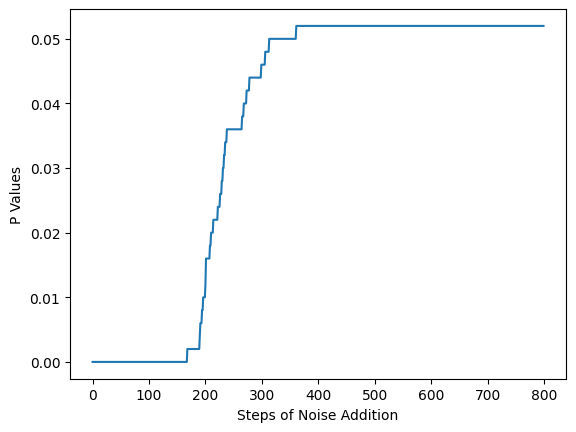}
    \caption{p values of replacement attack $p_w = 0.05$}
    \label{fig:beautiful}
\end{figure}

\begin{figure}[H]
    \centering
    \includegraphics[width=0.35\textwidth]{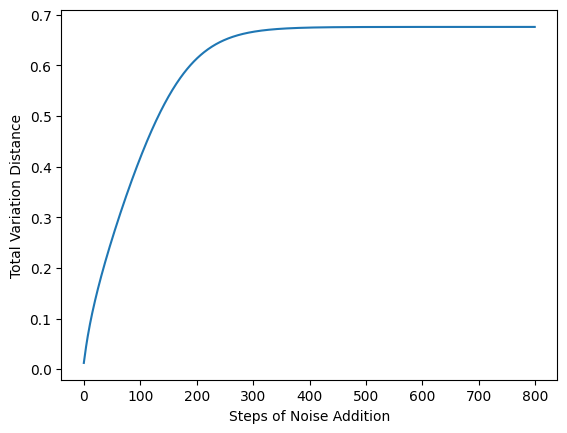}
    \caption{$d_{TV}(D, D_{attack,m}^{-1})$ of replacement attack $p_w = 0.05$}
    \label{fig:beautiful}
\end{figure}

\begin{figure}[H]
    \centering
    \includegraphics[width=0.35\textwidth]{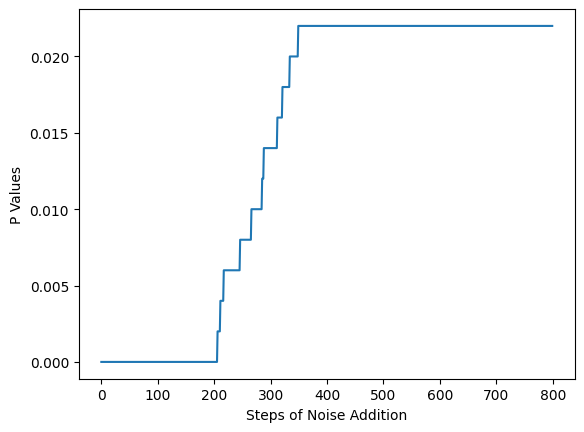}
    \caption{p values of replacement attack $p_w = 0.1$}
    \label{fig:beautiful}
\end{figure}

\begin{figure}[H]
    \centering
    \includegraphics[width=0.35\textwidth]{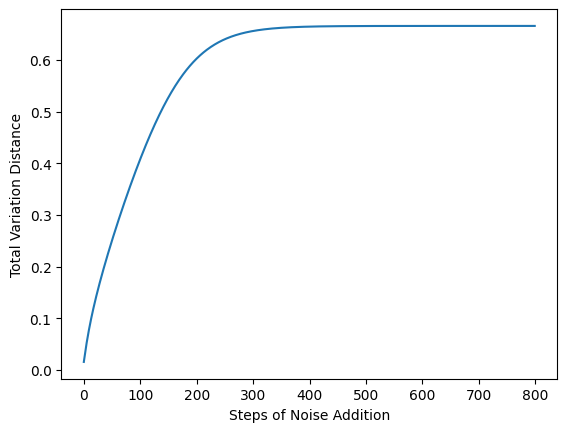}
    \caption{$d_{TV}(D, D_{attack,m}^{-1})$ of replacement attack $p_w = 0.1$}
    \label{fig:beautiful}
\end{figure}

Notice that in the first hundred iterations of the attack, the total variation distance between $D_{attack, m}^{-1}$ and $D$ remains relatively small. This demonstrates the robustness of our watermarking scheme. After 200 iterations, the total variation distance increases to $0.5$, suggesting that the detector should be cautious when concluding that the distribution is still an illicit copy. This observation is reasonable, as multiple iterations of the attack introduce more uniform noise into the distribution than the original signal. Consequently, our detector should avoid classifying a noisy distribution as an illicit copy.

\subsection{False Positive Rate}

If the data owner is confident in the effectiveness of the prior distribution, then by the nature of hypothesis testing, the false positive rate is bounded by the significance level. However, finding a single effective prior for all types of data is challenging. Using the same data of simulation 5, we evaluated the effectiveness of a Dirichlet distribution with a parameter vector where each element is set to $0.03$ and has a length equal to $M_D^{\secret}$. The detector's thresholds were set at $p_w = 0.05$ and $p_w = 0.1$. Instead of sampling probability vectors directly from the Dirichlet distribution, we constructed two different prior distributions, $Prior_1$ and $Prior_2$. In $Prior_1$, we sampled probability vectors by first creating a random vector of length $\#(D)$, where each element is independently sampled from a uniform distribution over the integers from $1$ to $16$. We then normalized this vector by dividing each element by its $L_1$ norm. In $Prior_2$, we followed the same process to construct a random vector, but then applied a softmax operation to convert the random vector into a probability vector. Then, we sample from these $1000$ probability vectors respectively from $Prior_1$ and $Prior_2$ and test the percentage of which the detector sets it to be positive, and we obtain the following result, letting the significance level to be $0.01$. 

\begin{table}[h]
    \centering
    \begin{tabular}{{|p{1.5cm}|p{1.7cm}|p{1.7cm}|}}
    \hline
    $p_w$ & $Prior_1$  FPR & $Prior_2$ FPR  \\
    \hline
    $0.05$ & $0.002$ & $0.017$ \\
    \hline 
    $0.1$ & $0.015$ & $0.009$ \\
    \hline
    \end{tabular}
    \caption{Simulation 6}
    \label{tab:sample_table}
\end{table}

The reason for setting the Dirichlet distribution's parameter to a small value is that smaller parameters yield greater variance when sampling probability vectors. If the sampled probability vectors are more dispersed on the simplex, we can better control the false positive rate. Given that the insertion inverse algorithm is highly effective, the detector is less likely to be compromised when using a prior distribution with high variance. Therefore, it is advisable for the data owner to test the false positive rate and design a robust prior distribution before selling watermarked data to different buyers. The data owner can even design a mixture distribution as prior, from which the probability vectors can evenly distributed on the probability simplex.

\section{Improvement}

In this section, we introduce two improved watermarking schemes based on the methods described in Section 3. The first method, called the sparse-column method, is designed to watermark a categorical distribution with a large number of categories. The second method, referred to as the pseudorandom mapping watermarking, is inspired by \cite{fu-etal-2024-gumbelsoft}. The pseudorandom mapping watermarking aims to reduce the fidelity damage caused by the watermarking scheme.

\subsection{Sparse-column method}

In Section 4.2, we observe that, when the number of samples in a table is fixed, the insertion inverse algorithm performs poorly if the number of categories is too large. Consequently, the total variation distance between $D_{inv}$ and $D$ increases, potentially leading to a poor true positive rate. Therefore, it is preferable to ignore some columns when inserting and detecting the watermark, which reduces the number of possible categories. Motivated by this, we designed the sparse-column method.

Unlike the original method, which splits $D$ column-wise into $X$ and $Y$, the sparse-column method divides $D$ into $tX$, $Z$, and $tY$. We assume that $Z$ consists of columns containing important information, which replacement attackers will avoid altering. Additionally, the mutual information between $Z$ and the other two distributions should be relatively high to ensure that data synthesizers do not disrupt the relationship between $Z$, $tX$, and $tY$.

Unlike the original method, which only takes $p_w$ and a single $\secret$ as parameters, the sparse-column method requires $p_w$, $\secret_1$, $\secret_2$, $\secret_3$, $x_{dim}$, and $y_{dim}$. Here, $p_w$ still represents the sparsity of the watermark, while $\secret_1$, $\secret_2$, and $\secret_3$ are strings used to construct three distinct hash functions.  Let $M$ and $N$ denote the dimensions of the distribution $tX$, with $x_{dim}$ being an integer value between $1$ and the dimension of $M$. Similarly, $y_{dim}$ is an integer between $1$ and the dimension of $N$. Later, we will see that $x_{dim}$ and $y_{dim}$ represent the number of columns that will be preserved in $tX$ and $tY$, respectively.

With these parameters, we develop the following XYExtraction Algorithm, from which we extract distributions $X$ and $Y$ from $tX$ and $tY$.

\begin{algorithm}[H]
    \caption{XYExtractor($T$; $\secret_1$, $\secret_2$, $x_{dim}$, $y_{dim}$)}
    \begin{algorithmic}[1]

    \STATE Split $T$ column-wise into $T_{tx}$, $T_{ty}$, and $T_z$
    \STATE Let $M$ denotes the number of column $T_{tx}$ has
    \STATE Let $N$ denotes the number of column $T_{ty}$ has
    \STATE $\alpha := \binom{M}{x_{dim}}$, $\beta:= \binom{N}{y_{dim}}$ 
    \STATE 
    \STATE Enumerate each combination of $x_{dim}$ columns in $T_{tx}$ from $0$ to $\alpha -1$
    \STATE Enumerate each combination of $y_{dim}$ columns in $T_{tx}$ from $0$ to $\beta -1$
    \STATE
    \FOR{Each sample $(tx, z, ty)$ in  $T$}
        \STATE Let $idx$ be the number $\Hash_{\alpha}^{\secret_1}(z)$ combination of columns in $tx$
        \STATE Let $idy$ be the nuumber $\Hash_{\beta}^{\secret_2}(z)$ combination of columns in $ty$
        \STATE let $tx[idx]$ be a sample for $T_x$ and $ty[idy]$ be a sample for $T_y$
    \ENDFOR 
    \STATE
    \STATE Combined $T_x$, $T_y$, and $T_z$ as $T_{simple}$
    \STATE Return $T_{simple}$
    \end{algorithmic}
\end{algorithm}

The rest of the watermark insertion scheme follows the exact same procedure with the regular method using $T_{simple}$ and $\secret_3$ instead of $T$ and $\secret$. Likewise, the detector needs to call the XYExtractor algorithm to extract a simple table to be tested.

\subsection{Pseudorandom Mapping Watermarking}
Note that our watermarking scheme preserves the distribution of $X$ while perturbing the distribution of $Y$. However, because of certain downstream tasks or applications that rely on the distribution of $Y$, some buyers may prefer a watermarking scheme that also preserves the marginal distribution of $Y$. In other words, we may need a watermarking scheme that maintains the marginal distributions of both $X$ and $Y$, while only altering the relationship between them. It turns out that we can achieve this goal with a slight modification to our ordinary watermarking scheme from Section 3.

Consider a pseudorandom generator $\Gen_{\secret}(\cdot)$ that outputs a number in $[0,1]$, appearing as if it follows a continuous uniform distribution without access to the key.
Specifically, $\Gen_{\secret}(x)$ uses $\Hash^{\secret}(x)$ as the seed when generating the output. Note that here, the output of $\Hash^{\secret}(x)$ is a binary sequence, unlike the integer format used in previous sections. Let $[y_0, y_1, \dots, y_n]$ be the probability vector of $Y$.

Then, we define $\AdvHash_{\secret}(x):= M_Y^{-1}(k)$ if $\Gen_{\secret}(x) \in [\sum_{i= 0}^{k-1}y_i, \sum_{i=0}^{k}y_i) $. Then, the pseudorandom mapping watermarking follows the exact same procedure with the regular watermarking scheme but replacing all $\Hash_{\#(Y)}^{\secret}(\cdot)$ with $\AdvHash_{\secret}(\cdot)$.

Notice that if we let $\tilde{Y}_{\secret} = \AdvHash_{\secret}(X)$ has the exact same marginal distribution with $Y$. The probability of a $\Gen_{\secret}(x) \in [\sum_{i= 0}^{k-1}y_i, \sum_{i=0}^{k}y_i)$ is equal to $y_k$ for all $k$. Thus, the marginal distributions of $\tilde{Y}_{\secret}$ and $\hat{Y}_{\secret}$ are exactly same with $Y$.

\section{Conclusion}
 Our study introduces a novel distribution-level watermarking framework for generative categorical data, addressing the challenges unique to the era of synthetic data generation. By inserting watermarks systematically into the underlying data distribution and employing robust statistical verification techniques, our method ensures high reliability. The empirical and theoretical validations highlight the framework's robustness against transformations commonly encountered in generative models, such as data regeneration and attacks on watermark integrity. Moreover, we have extended our approach with two advanced methods: the sparse-column method, which improves performance for datasets with large categorical spaces, and the pseudorandom mapping watermarking technique, which preserves the marginal distributions of the data while maintaining watermark reliability. This contribution lays a foundation for advancing trustworthy and traceable synthetic data generation, aligning with the growing need for transparency and accountability in AI-driven systems.

\ifCLASSOPTIONcompsoc
 
\else
  
\fi

\bibliographystyle{IEEEtran}
\bibliography{ref}

\end{document}